\documentclass[10pt,oneside,english]{amsart}
\usepackage[T1]{fontenc}
\usepackage[latin9]{inputenc}
\usepackage[letterpaper]{geometry}
\usepackage{amsthm}
\usepackage{amstext}
\usepackage{amssymb}
\usepackage{esint}

\makeatletter
\numberwithin{equation}{section}
\numberwithin{figure}{section}
\theoremstyle{plain}
\newtheorem{thm}{Theorem}
  \theoremstyle{definition}
  \newtheorem{defn}[thm]{Definition}
  \theoremstyle{plain}
  \newtheorem{prop}[thm]{Proposition}
  \theoremstyle{plain}
  \newtheorem*{thm*}{Theorem}

\makeatother

\makeatother

\usepackage{babel}

\begin{document}

\title{A probabilistic Approach to some results by Nieto and Truax}

\author{C. Vignat}

\address{Laboratoires des Signaux et Systèmes, Université d'Orsay, France}
 \email{christophe.vignat@lss.supelec.fr}
\maketitle

\section{Introduction}

In this paper, we revisit some results by Nieto and Truax about generating
functions for arbitrary order coherent and squeezed states. These
results were obtained using the exponential of the Laplacian operator;
they were later extended by Dattoli et al. \cite{Dattoli} using more
elaborated operational identities. In this paper, we show that the
operational approach can be replaced by a simple probabilistic approach,
in the sense that the exponential of derivatives operators can be
replaced by equivalent expectation operators. This approach brings
new insight about the links between operational and probabilistic
calculus. 

In the first part, we show that the exponential of the derivation
operator of arbitrary integer order can be replaced by an expectation
with repect to a carefuly chosen random variable.

In the second part, we apply this result to the Gould-Hopper polynomials,
which include as special cases the Kampé de Fériet and the Hermite
polynomials; this allows us to recover and interpret easily some properties
of these polynomials. 

The last part is dedicated to the application of these probabilistic
representations to the computation of coherent and squeezed states.

\section{The operator $\exp\left[\left(c\frac{d}{dx}\right)^{j}\right]$\label{sec:operator}}

\subsection{Introduction}

In \cite{Nieto}, Nieto and Truax consider the operators
\[
I_{j}=\exp\left[\left(c\frac{d}{dx}\right)^{j}\right]
\]
 where $c$ is a constant and $j$ an integer. They remark that, for
any well-behaved function$f,$ $I_{1}$ acts as the translation operator\[
I_{1}f\left(x\right)=f\left(x+c\right).\]
 Moreover, $Z_{2}$ being a Gaussian random variable with variance
$2$, $I_{2}$ acts as the Gauss-Weierstrass transform\[
I_{2}f\left(x\right)=E_{Z_{2}}f\left(x+cZ_{2}\right).\]
 Since $I_{1}f$ can be also written as $E_{Z_{1}}f\left(x+Z_{1}\right)$
where $Z_{1}$ is the deterministic variable equal to $1,$ it is
tempting to wonder if the expression\[
I_{j}=E_{Z_{j}}f\left(x+cZ_{j}\right)\]
 holds for values of $j\ge3$ and to study what random variable $Z_{j}$
possibly comes out in this formula. In \cite{Nieto}, the following
general expression is proposed\begin{eqnarray}
I_{j}f\left(x\right) & = & \frac{1}{\left(2\pi\right)^{\frac{j-1}{2}}\sqrt{j}}\int_{0}^{+\infty}\frac{dx_{1}e^{-x_{1}}}{x_{1}^{\frac{1}{j}}}\int_{0}^{+\infty}\frac{dx_{2}e^{-x_{2}}}{x_{2}^{\frac{2}{j}}}\dots\int_{0}^{+\infty}\frac{dx_{j-1}e^{-x_{j-1}}}{x_{j-1}^{\frac{j-1}{j}}}\label{eq:Im}\\
 & \times & \sum_{l=1}^{j}f\left(x+jc\left(x_{1}x_{2}\dots x_{j-1}\right)^{\frac{1}{j}}\exp\left(2i\pi\frac{l}{j}\right)\right)\nonumber \end{eqnarray}
 and the authors add: {}``This result, although in closed form, is
usually too complicated to yield results in terms of elementary functions''. The next section is devoted to providing a simple expression for the operator $I_{j}.$

\subsection{a complex valued stable random variable}

We first introduce the following random variables $W_{j}$ and $Z_{j}$
where in this whole section, $j$ is an integer such that $j\ge2.$
\begin{defn}
For $j\in\mathbb{N},$ $j\ge2,$ define the root of the unity \[
\omega_{j}=\exp\left(\imath\frac{2\pi}{j}\right)\]
 and the random variable $W_{j}$ as\begin{equation}
\Pr\left\{ W=\omega_{j}^{l}\right\} =\frac{1}{j},\,\,\,0\le l\le j-1.\label{eq:W}\end{equation}
 Thus $W_{j}$ equals equiprobably all values of the $j-$roots of
the unity. The random variable $Z_{j}$ is then defined as \[
Z_{j}=W_{j}S_{\frac{1}{j}}^{-\frac{1}{j}}\]
 where $S_{\frac{1}{j}}$ is a stable random variable \cite{Feller} with characteristic  parameter
$\frac{1}{j}$, independent of $W_{j}.$

In the following, when possible, we'll discard the indices for notational
simplicity and write simply\begin{equation}
Z=WS^{-\frac{1}{j}}.\label{eq:Z}\end{equation}

We now describe some properties of the random variable $Z.$\end{defn}
\begin{prop}
The random variable $Z$ has moments\begin{equation}
EZ^{k}=\begin{cases}
0 & \text{if}\,\, k\ne pj\\
\frac{\left(pj\right)!}{p!} & \text{if}\,\, k=pj,\, p\in\mathbb{N}.\end{cases}\label{eq:Zmoments}\end{equation}
 \end{prop}
\begin{proof}
The $k-th$ order moment of $Z$ is\[
EZ^{k}=EW^{k}E\frac{1}{S^{\frac{k}{j}}}\]
 with \[
EW^{k}=\frac{1}{j}\sum_{l=0}^{j-1}\omega_{j}^{kl}=\begin{cases}
0 & \text{if}\,\, k\ne pj\\
1 & \text{if}\,\, k=pj,\, p\in\mathbb{N}\end{cases}\]

Moreover, with $k=pj,$\[
E\frac{1}{S^{\frac{k}{j}}}=E\frac{1}{S^{p}}=\frac{\left(pj\right)!}{p!}\]
 (see \cite{Williams}) so that the result follows. 
\end{proof}
As a consequence, the generating function $\varphi_{Z}\left(u\right)=E\exp\left(uZ\right)$
of $Z$ can be evaluated as follows. 
\begin{prop}
The generating function of $Z$ is\[
\varphi_{Z}\left(u\right)=\exp\left(u^{j}\right),\,\,\, u\ge0.\]
\end{prop}
\begin{proof}
A straightforward computation gives\[
\varphi_{Z}\left(u\right)=\sum_{k=0}^{+\infty}\frac{u^{k}}{k!}EZ^{k}=\sum_{p=0}^{+\infty}\frac{u^{pj}}{pj!}\frac{pj!}{p!}=\exp\left(u^{j}\right).\]

\end{proof}
From this property, we deduce that the random variable $Z$ exhibits
the following stability property. 
\begin{prop}
If $Z_{1}$ and $Z_{2}$ are two independent random variables distributed
as in (\ref{eq:Z}) with parameter $j$ and if $a_{1}$ and $a_{2}$
are two real positive numbers then the random variable\[
a_{1}Z_{1}+a_{2}Z_{2}\sim\left(a_{1}^{j}+a_{2}^{j}\right)^{\frac{1}{j}}Z\]
where $Z$ is again distributed as in (\ref{eq:Z}) with same parameter
$j$ and $\sim$ denotes equality in distribution. 
\end{prop}
We note that $Z$ is a complex valued random variable; for real valued
random variables, the stability property can hold only for values
of the parameter $j$ such that $0<j\le2$ (see \cite[p. 170]{Feller}).

\subsection{application to the operator $\exp\left[\left(c\frac{d}{dx}\right)^{j}\right]$}

We can now give a simplified version of the formula (\ref{eq:Im})
as follows. 
\begin{thm}
The operator $I_{j}=\exp\left[\left(c\frac{d}{dx}\right)^{j}\right]$
acts on a function $f$ as\begin{equation}
I_{j}f\left(x\right)=E_{Z_{j}}f\left(x+cZ\right)\label{eq:Imsimplified}\end{equation}
 where $Z_{j}$ is defined as in (\ref{eq:Z}), for any function $f$
such that the right-hand side of this equality exists.\end{thm}
\begin{proof}
By definition, with $Z$ defined as in (\ref{eq:Z}), \[
I_{j}f\left(x\right)=\exp\left[\left(c\frac{d}{dx}\right)^{j}\right]f\left(x\right)=E_{Z_{j}}\exp\left(cZ\frac{d}{dx}\right)f\left(x\right)=E_{Z_{j}}f\left(x+cZ\right).\]

\end{proof}
We note that the link between equalities (\ref{eq:Im}) and (\ref{eq:Imsimplified})
can be identified using Williams' formula \cite[end of part 2]{Williams}:
formula (\ref{eq:Im}) can indeed be rewritten as\[
I_{j}f\left(x\right)=E_{W}E_{Y_{\frac{1}{j}},Y_{\frac{2}{j}},\dots,Y_{\frac{j-1}{j}}}f\left(x+cW\left(Y_{\frac{1}{j}}Y_{\frac{2}{j}}\dots Y_{\frac{j-1}{j}}\right)^{\frac{1}{j}}\right)\]
 where $Y_{\frac{1}{j}},Y_{\frac{2}{j}},\dots,Y_{\frac{j-1}{j}}$
is the product of $j-1$ independent Gamma random variables with respective
shape parameters $\frac{1}{j},\frac{2}{j},\dots,\frac{j-1}{j}.$ But
by Williams' formula,\[
Y_{\frac{1}{j}}Y_{\frac{2}{j}}\dots Y_{\frac{j-1}{j}}\sim\left(\frac{1}{j}\right)^{j}S_{\frac{1}{j}}^{-1}\]
where $S_{\frac{1}{j}}$ is a stable random variable with parameter
$\frac{1}{j}$, so that \[
I_{j}f\left(x\right)=E_{W,S_{\frac{1}{j}}}f\left(x+cWS_{\frac{1}{j}}^{-\frac{1}{j}}\right)\]
which is exactly (\ref{eq:Imsimplified}).


As an example, we consider the case $j=2$ for which $W$
is Bernoulli distributed with parameter $\frac{1}{2}$ and since
$S_{\frac{1}{2}}$ is Lévy distributed, according to \cite[p.456]{Devroye},
\[
\frac{1}{\sqrt{S_{\frac{1}{2}}}}\sim\sqrt{2}\vert N\vert\]
where $N$ is Gaussian with variance $\sigma_{N}^{2}=1$ so that $Z_{2}=W_{2}S_{\frac{1}{2}}^{-\frac{1}{2}}$
is itself Gaussian with variance $\sigma_{Z}^{2}=2$ and we recover
the Gauss-Weierstrass transform operator \[
I_{2}f\left(x\right)=Ef\left(x+cZ_{2}\right).\]

\section{Properties of the Gould Hopper polynomials}

As an application of the former result, we consider the particular
case $f\left(x\right)=x^{n}.$ The polynomials \begin{equation}
g_{n}^{m}\left(x,h\right)=\exp\left[h\left(\frac{d}{dx}\right)^{m}\right]x^{n},\,\,\,\left(m,n\right)\in\mathbb{N}^{2}\label{eq:GouldHopper}\end{equation}
are known as the Gould-Hopper polynomials, as introduced in \cite[p.58]{Gould}.
The particular case $m=2$ corresponds to the Kampé de Fériet polynomials,
and the polynomials $g_{n}^{2}\left(2x,-1\right)$ coincide with the
classical Hermite polynomials. From the results of section \ref{sec:operator}
and with $h=c^{m}$, we deduce\begin{equation}
g_{n}^{m}\left(x,h\right)=E_{Z}\left(x+h^{\frac{1}{m}}Z\right)^{n}\label{eq:gnm}\end{equation}
where $Z$ is distributed as in (\ref{eq:Z}) with $j=m$ \footnote{we replace here $j$ by $m$ to stick to the notation introduced by Gould}. We note
that this representation is different from the one given by \cite[Thm.1]{Nadarajah},
which seems of little use since it involves integer moments of stable
random variables.

We show now that the representation (\ref{eq:gnm}) allows to recover
and extend easily some well-known properties of the GH polynomials.
\begin{enumerate}
\item the generating function of the GH poynomial is \cite[(6.3)]{Gould}\begin{eqnarray*}
\sum_{n=0}^{+\infty}g_{n}^{m}\left(x,h\right)\frac{t^{n}}{n!} & = & E_{Z}\sum_{n=0}^{+\infty}\left(x+h^{\frac{1}{m}}Z\right)^{n}\frac{t^{n}}{n!}\\
 & = & E_{Z}\exp\left(t\left(x+h^{\frac{1}{m}}Z\right)\right)=\exp\left(tx+ht^{m}\right);\end{eqnarray*}

\item the derivative of a GH polynomial \cite[(6.4)]{Gould} is easily computed
as \[
\frac{d}{dx}g_{n}^{m}\left(x,h\right)=nE_{Z}\left(x+h^{\frac{1}{m}}Z\right)^{n-1}=ng_{n-1}^{m}\left(x,h\right).\]

\item the GH polynomials satisfy the following addition theorem\begin{equation}
g_{n}^{m}\left(\sum_{i=1}^{r}x_{i},\sum_{i=1}^{r}h_{i}\right)=\sum_{n_{1}+\dots+n_{r}=n}\prod_{k=1}^{r}g_{n_{k}}^{m}\left(x_{k},h_{k}\right),\label{eq:addition Gould}\end{equation}
the proof of which is straightforward using the representation (\ref{eq:gnm}).
This result generalizes the case $r=2$ obtained by Gould and Hopper
under the form \cite[(6.19)]{Gould}\[
\sum_{k=0}^{n}\binom{n}{k}g_{k}^{m}\left(x,h\right)g_{n-k}^{m}\left(y,h\right)=g_{n}^{m}\left(x+y,2h\right).\]
Moreover, the classical addition theorem for Hermite polynomials \cite[(40) p. 196]{Erdelyi}\[
\frac{\left(\sum_{k=1}^{r}a_{k}^{2}\right)^{\frac{n}{2}}}{n!}H_{n}\left(\frac{\sum_{k=1}^{r}a_{k}x_{k}}{\sqrt{\sum_{k=1}^{r}a_{k}^{2}}}\right)=\sum_{n_{1}+\dots+n_{r}}\prod_{k=1}^{r}\frac{a^{m_{k}}}{m_{k}!}H_{m_{k}}\left(x_{k}\right)\]
is easily deduced from (\ref{eq:addition Gould}) using the link between
Hermite polynomials and Gould polynomials mentioned above. 
\item Another result from \cite[(6.9)]{Gould}, namely\[
\exp\left(h\frac{d^{j}}{dx^{j}}\right)\left(x^{n}\exp\left(tx\right)\right)=\frac{d^{n}}{dt^{n}}\left(\exp\left(tx+ht^{j}\right)\right)\]
can be recovered easily using (\ref{eq:Imsimplified}) since\begin{eqnarray*}
\exp\left(h\frac{d^{j}}{dx^{j}}\right)\left(x^{n}\exp\left(tx\right)\right) & = & E_{Z}\left(x+h^{\frac{1}{j}}Z\right)^{n}\exp\left(t\left(x+h^{\frac{1}{j}}Z\right)\right)\\
 & = & \frac{d^{n}}{dt^{n}}E_{Z}\exp\left(t\left(x+h^{\frac{1}{j}}Z\right)\right)=\frac{d^{n}}{dt^{n}}\left(\exp\left(tx+ht^{j}\right)\right).\end{eqnarray*}

\end{enumerate}

\section{Application to Coherent and Squeezed States}

\subsection{Multisection of series}

In calculations of higher-order and squeezed states quantities, the
following generating functions of Hermite polynomials appear - we
follow here the notation by Nieto and Truax:\begin{eqnarray}
S\left(j,k,z\right) & = & \sum_{n=0}^{+\infty}\frac{z^{jn+k}}{\left(jn+k\right)!}\label{eq:S(j,k,z)}\\
G\left(j,k,x,z\right) & = & \sum_{n=0}^{+\infty}\frac{z^{jn+k}H_{jn+k}\left(x\right)}{\left(jn+k\right)!}\label{eq:G(j,k,x,z)}\end{eqnarray}
and\begin{equation}
g\left(j,k,x,y,z\right)=\sum_{n=0}^{+\infty}\frac{z^{jn+k}H_{jn+k}\left(x\right)H_{jn+k}\left(y\right)}{\left(jn+k\right)!}\label{eq:g(j,k,x,y,z)}\end{equation}
 where $H_{n}\left(x\right)$ is the Hermite polynomial of degree
$n,$ and $j$ and $k$ are integers.

These quantities are all computed in \cite{Nieto} using an operational
calculus approach, but we show here that the probabilistic approach
introduced in Section \ref{sec:operator} gives a new insight on these
formulas and allows to extend them to other classes of polynomials.

We introduce the following theorem.
\begin{thm}
\label{thm:general1}If \[
\phi\left(t,x\right)=\sum_{n=0}^{+\infty}t^{n}f_{n}\left(x\right)\]
is the generating function of the sequence of functions $\left\{ f_{n}\left(x\right)\right\} $
then %
\footnote{we recall the notation $E_{W_{j}}f\left(W\right)=\frac{1}{j}\sum_{l=0}^{j-1}f\left(\omega_{j}^{l}\right)$%
}\[
\sum_{n=0}^{+\infty}t^{jn+k}f_{jn+k}\left(x\right)=E_{W_{j}}W^{-k}\phi\left(tW,x\right)\]
where $W$ is distributed as in (\ref{eq:W}).\end{thm}
\begin{proof}
We compute \[
E_{W_{j}}W^{-k}\phi\left(tW,x\right)=\sum_{n=0}^{+\infty}E_{W_{j}}W^{n-k}t^{n}f_{n}\left(x\right)\]
and use the property (\ref{eq:Zmoments}) to deduce the result.
\end{proof}
We note that this theorem is the probabilistic formulation of the
technique of {}``multisection of series'' described by Riordan \cite[section 4.3]{Riordan}
as a {}``process of ancient vintage''. This process is also characterized
in \cite[(19)]{Dattoli} as a consequence of the sieving principle.

\subsection{Applications}

As an application of this theorem, we recover easily the results in
\cite{Nieto}:
\begin{enumerate}
\item choosing $f_{n}\left(x\right)=\frac{1}{n!}\,\,\forall x\in\mathbb{R},\,\,\forall n\in\mathbb{N}$
we deduce\[
\phi\left(t,x\right)=\exp\left(t\right)\]
and\[
S\left(j,k,z\right)=E_{W_{j}}W^{-k}\exp\left(zW\right)\]

\item choosing $f_{n}\left(x\right)=\frac{H_{n}\left(x\right)}{n!},$ we
deduce the well-known generating function of the Hermite polynomials\[
\phi\left(t,x\right)=\exp\left(2tx-t^{2}\right)\]
so that \[
G\left(j,k,x,z\right)=E_{W_{j}}W^{-k}\exp\left(2Wxz-z^{2}W^{2}\right)\]

\item the more general case \[
\sum_{n=0}^{+\infty}\frac{t^{jn+k}}{\left(jn+k\right)!}H_{jn+k+m}\left(x\right)\]
can also be easily derived: choosing $f_{n}\left(x\right)=\frac{H_{n+m}\left(x\right)}{n!},$
we deduce the generating function \cite[(1) p.197]{Rainville}\[
\phi\left(t,x\right)=\exp\left(2tx-t^{2}\right)H_{m}\left(x-t\right)\]
so that \[
\sum_{n=0}^{+\infty}\frac{t^{jn+k}}{\left(jn+k\right)!}H_{jn+k+m}\left(x\right)=E_{W_{j}}W^{-k}\exp\left(2tWx-t^{2}W^{2}\right)H_{m}\left(x-tW\right)\]
which coincides with \cite[(21)]{Dattoli}.
\item remarking that the variable $x$ in theorem (\ref{thm:general1})
may be multidimensional, choosing $x=\left(x_{1},x_{2}\right)$ and
\[
f\left(x\right)=\frac{H_{n}\left(x_{1}\right)H_{n}\left(x_{2}\right)}{n!}\]
we deduce, by the Mehler formula \cite[(22) p.194]{Erdelyi},\[
\phi\left(t,x_{1},x_{2}\right)=\frac{1}{\sqrt{1-4t^{2}}}\exp\left(-4\frac{t^{2}x^{2}+t^{2}y^{2}-txy}{1-4t^{2}}\right)\]
so that \[
g\left(j,k,x,y,z\right)=E_{W_{j}}W^{-k}\frac{1}{\sqrt{1-4t^{2}W^{2}}}\exp\left(-4\frac{t^{2}W^{2}x^{2}+t^{2}W^{2}y^{2}-tWxy}{1-4t^{2}W^{2}}\right).\]

\end{enumerate}
We remark that Theorem \ref{thm:general1} is not restricted to Hermite
polynomials; let us give another few examples:
\begin{enumerate}
\item for the Gould Hopper polynomials $g_{n}^{m}\left(x,h\right)$, choosing
$f_{n}\left(x\right)=\frac{g_{n}^{m}\left(x,h\right)}{n!}$ as defined
in \eqref{eq:GouldHopper}, we deduce\[
\sum_{n=0}^{+\infty}\frac{t^{jn+k}}{\left(jn+k\right)!}g_{jn+k}\left(x,h\right)=E_{W_{j}}W^{-k}\exp\left(txW+ht^{m}W^{m}\right)\]

\item for the Laguerre polynomials, with $f_{n}\left(x\right)=L_{n}^{\alpha}\left(x\right)$
we have $\phi\left(t,x\right)=\left(1-t\right)^{-\alpha-1}\exp\left(\frac{xt}{t-1}\right)$
so that\[
\sum_{n=0}^{+\infty}\frac{t^{jn+k}}{\left(jn+k\right)!}L_{jn+k}^{\alpha}\left(x\right)=E_{W_{j}}W^{-k}\left(1-tW\right)^{-\alpha-1}\exp\left(\frac{xtW}{tW-1}\right).\]

\end{enumerate}

\subsection{Further sums}

In their study, Nieto and Truax consider also the sums \begin{equation}
K\left(j,k,p,q,x,z\right)=\sum_{n=0}^{+\infty}\frac{z^{jn+k}H_{jn+k}\left(x\right)}{\left(pn+q\right)!}\label{eq:K(j,k,p,q,x,z)}\end{equation}
for which they explicit some specific cases. We give here the general
result only in the cases $j=p$ and $j=2p$ by looking first at the
simple sum\begin{equation}
k\left(j,k,p,q,z\right)=\sum_{n=0}^{+\infty}\frac{z^{jn+k}}{\left(pn+q\right)!}.\label{eq:k(j,j,p,q,z)}\end{equation}

A straightforward computation shows that this sum is related to the
sum $S\left(j,k,z\right)$ in (\ref{eq:S(j,k,z)}) as\begin{equation}
k\left(j,k,p,q,z\right)=z^{k-j\frac{q}{p}}S\left(p,q,z^{\frac{j}{p}}\right).\label{eq:k}\end{equation}

As a generalization of this result, we obtain the following
\begin{thm*}
If the sequence of functions $\left\{ h_{n}\left(x\right)\right\} $
can be expressed as a sequence of moments \[
h_{n}\left(x\right)=E_{U}\varphi^{n}\left(x,U\right)\]
for some function $\varphi,$ then\begin{equation}
\sum_{n=0}^{+\infty}\frac{z^{jn+k}}{\left(pn+q\right)!}h_{jn+k}\left(x\right)=E_{W_{p},U}\left\{ \left(z\varphi\left(x,U\right)\right)^{k-j\frac{q}{p}}W^{-q}\exp\left(\left(z\varphi\left(x,U\right)\right)^{\frac{j}{p}}W\right)\right\} .\label{eq:jnkpnq}\end{equation}
\end{thm*}
\begin{proof}
The proof is straightforward replacing $z$ by $z\varphi\left(x,U\right)$
in (\ref{eq:k}) and taking expectation over $U.$
\end{proof}
As a consequence, using the moment representation of the Hermite polynomials\[
H_{n}\left(x\right)=2^{n}E_{N}\left(x+iN\right)^{n}\]
where $N$ is a Gaussian random variable with variance $\sigma_{N}^{2}=\frac{1}{2},$
we deduce

\begin{equation}
K\left(j,k,p,q,x,z\right)=\left(2z\right)^{k-j\frac{q}{p}}E_{W_{p},N}W^{-q}\left(x+iN\right)^{k-j\frac{q}{p}}\exp\left(\left(2z\left(x+iN\right)\right)^{\frac{j}{p}}W\right)\label{eq:jnkpnqHermite}\end{equation}
The expectation with respect to the Gaussian variable $N$ is difficult
to obtain for an arbitrary value of the ratio $\frac{j}{p}$ so that
in the following, we give explicit expressions only in the cases $j=p$
and $j=2p.$

\subsubsection{the case $j=p$}

In this case we obtain the following result.
\begin{thm}
If $k-q\in\mathbb{N},$ the function $K\left(j,k,j,q,x,z\right)$
reads\begin{eqnarray*}
K\left(j,k,j,q,x,z\right) & = & z^{k-q}E_{W_{j}}W^{-q}\exp\left(-zW\left(zW-2x\right)\right)H_{k-q}\left(x-zW\right)\end{eqnarray*}
\end{thm}
\begin{proof}
The proof is immediate remarking that in (\ref{eq:jnkpnqHermite})
the Gaussian expectation \begin{eqnarray*}
E_{N}\left(x+iN\right)^{k-q}\exp\left(2zW\left(x+iN\right)\right) & = & \left(2W\right)^{q-k}\frac{d^{k-q}}{dz^{k-q}}\exp\left(2xzW\right)E_{N}\exp\left(i2zWN\right)\end{eqnarray*}
involves the Gaussian characteristic function $E_{N}\exp\left(i2zWN\right)=\exp\left(-z^{2}W^{2}\right)$
so that it is equal to\[
\left(2W\right)^{q-k}\exp\left(x^{2}\right)\frac{d^{k-q}}{dz^{k-q}}\exp\left(-\left(zW-x\right)^{2}\right)\]
which, using Rodriguez formula for the Hermite polynomials, can be
expressed as\[
2^{q-k}\exp\left(x^{2}\right)\exp\left(-\left(zW-x\right)^{2}\right)\left(-1\right)^{k-q}H_{k-q}\left(zW-x\right)\]
from which we deduce \[
K\left(j,k,j,q,x,z\right)=z^{k-q}E_{W_{j}}W^{-q}\exp\left(-zW\left(zW-2x\right)\right)H_{k-q}\left(x-zW\right)\]

\end{proof}

\subsubsection{case $j=2p$}

This case can also be solved as follows.
\begin{thm}
If $k-2q\in\mathbb{N},$ the function $K\left(2p,k,p,q,x,z\right)$
reads\begin{eqnarray*}
K\left(2p,k,p,q,x,z\right) & = & z^{k-2q}E_{W_{p}}W^{-q}\frac{\exp\left(\frac{4z^{2}W}{1+4z^{2}W}\right)}{\left(1+4z^{2}W\right)^{\frac{k-2q+1}{2}}}H_{k-2q}\left(\frac{x}{\sqrt{1+4z^{2}W}}\right)\end{eqnarray*}
\end{thm}
\begin{proof}
With $j=2p,$ replacing $x$ by $ix$, we need to compute the Gaussian
expectation\[
E_{N}\left(x+N\right)^{k-2q}\exp\left(-4z^{2}W\left(x+N\right)^{2}\right)=\frac{1}{\sqrt{\pi}}\int_{-\infty}^{+\infty}\left(x+N\right)^{k-2q}\exp\left(-4z^{2}W\left(x+N\right)^{2}\right)\exp\left(-N^{2}\right)dN\]
which, after change of variable $M=x+N$ reads\[
\frac{1}{\sqrt{\pi}}\exp\left(-x^{2}\right)\int_{-\infty}^{+\infty}M^{k-2q}\exp\left(-M^{2}\left(1+4z^{2}W\right)\right)\exp\left(2xM\right)dM.\]
This classical Gaussian integral is equal to \cite[2.3.15.10]{Prudnikov}\[
\left(\frac{i}{2}\right)^{k-2q}\exp\left(-x^{2}\right)\frac{\exp\left(\frac{x^{2}}{1+4z^{2}W}\right)}{\left(1+4z^{2}W\right)^{\frac{k-2q+1}{2}}}H_{k-2q}\left(\frac{ix}{\sqrt{1+4z^{2}W}}\right)\]
and we deduce, replacing $x$ by $\left(-ix\right)$,\[
K\left(2p,k,p,q,x,z\right)=z^{k-2q}E_{W_{p}}W^{-q}\frac{\exp\left(\frac{4z^{2}W}{1+4z^{2}W}x^{2}\right)}{\left(1+4z^{2}W\right)^{\frac{k-2q+1}{2}}}H_{k-2q}\left(\frac{x}{\sqrt{1+4z^{2}W}}\right).\]

\end{proof}




\section{Conclusion}

In this paper, we have shown that there exists a probabilistic counterpart
to the classical operational calculus. A further direction of research
is the extension of this approach to the multivariate case.


\begin{thebibliography}{12}
\bibitem{Gould}H. W. Gould and A. T. Hopper, Operational formulas
connected with two generalizations of Hermite polynomials, Duke Math.
J. 29-1 (1962), 51-63

\bibitem{Nieto}M. M. Nieto and D. R. Truax, Arbitrary-order Hermite
generating functions for obtaining arbitrary-order coherent and squeezed
states, Physics Letters A 208 (1995), 8-16

\bibitem{Williams}E.J. Williams, Some representations of stable random
variables as products, Biometrika 64-1 (1977), 167-169

\bibitem{Devroye}L. Devroye, Non-Uniform Random Variate Generation,
Springer-Verlag, (1986)

\bibitem{Nadarajah}S. Nadarajah, Simple formulas for certain polynomials,
Applied Mathematics and Computation 187 (2007), 1592\textendash{}1596

\bibitem{Feller}W. Feller, An Introduction to Probability Theory
and Its Applications, Volume 2, Wiley, 1971

\bibitem{Erdelyi}H. Bateman and A. Erdélyi, Higher transcendental
functions, Volume 2, McGraw-Hill, 1953

\bibitem{Dattoli}G. Dattoli, A. Torre and M. Carpanese, Operational
Rules and Arbitrary Order Hermite Generating Functions, Journal of
mathematical analysis and applications, 227-1 (1998), 98-111

\bibitem{Riordan}J. Riordan, Combinatorial Identities, Robert E.
Krieger Publishing Company, 1979

\bibitem{Prudnikov}A.P. Prudnikov, Yu.A. Brychkov, O.I. Marichev,
Integrals and Series Volume 1: Elementary Functions, CRC Press, 1986


\bibitem{Rainville}E.D. Rainville, Special Functions, The MacMillan
Company, 1960
\end{thebibliography}
\end{document}